\documentclass[11pt]{article}
\usepackage{bbold} 
\usepackage{amsmath}
\usepackage{amssymb}
\usepackage{amsthm}
\usepackage[english]{babel}
\usepackage[utf8]{inputenc}
\usepackage[T1]{fontenc}
\usepackage{lmodern}
\usepackage{microtype}
\usepackage{enumerate}
\usepackage{mathrsfs}
\usepackage{mathtools}
\usepackage{graphicx}
\usepackage{subfigure}
\usepackage{float}
\usepackage{color}
\usepackage{csquotes}
\usepackage{mathrsfs}
\usepackage{hyperref}

\usepackage[a4paper, top=2.5cm, bottom=2.5cm, left=1.5cm, right=1.5cm]{geometry}

\usepackage{cite}

\theoremstyle{plain}
\newtheorem{teo}{Theorem}[section]
\newtheorem{lem}[teo]{Lemma}

\newtheorem*{teo*}{Teorema}

\theoremstyle{remark}
\newtheorem{oss}{Remark}

\theoremstyle{definition}

\newtheorem*{defi*}{Definition}

\renewcommand{\phi}{\varphi}
\renewcommand{\epsilon}{\varepsilon}

\renewcommand{\bar}{\overline}

\newcommand{\EE}{\mathbb{E}}
\newcommand{\QQ}{\mathbb{Q}}
\newcommand{\RR}{\mathbb{R}}

\newcommand{\FF}{\mathscr{F}}
\newcommand{\GG}{\mathscr{G}}
\newcommand{\HH}{\mathscr{H}}
\renewcommand{\AA}{\mathscr{A}}

\newcommand{\F}{\mathcal{F}}

\newcommand{\II}[1]{{1}_{\{ #1\}}}

\title{{\bf \Large Invariance, existence and uniqueness of solutions of
    nonlinear valuation PDEs and FBSDEs inclusive of credit risk, collateral and
    funding costs}\thanks{The opinions here expressed are solely those of the
    authors and do not represent in any way those of their employers.
    Ackwnowledgements. We are grateful to Cristin Buescu, Jean-Fran\c{c}ois
    Chassagneux, François Delarue and Marek Rutkowski for helfpul discussion and suggestions that helped us improve the paper. Marek Rutkowski and Andrea Pallavicini visits were funded via the EPSRC Mathematics Platform grant EP/I019111/1.}   }
\author{
Damiano Brigo\thanks{Dept. of Mathematics, Imperial College London {\tt damiano.brigo@imperial.ac.uk}}
\ \ \ Marco Francischello\thanks{Dept. of Mathematics, Imperial College London, {\tt m.francischello14@imperial.ac.uk}}  \ \ 
Andrea Pallavicini\thanks{Imperial College London and Banca IMI Milan, {\tt a.pallavicini@imperial.ac.uk}}
}
\date{
\small First Version: February 1, 2014.  This version: \today
}

\begin{document}

\maketitle

\begin{abstract}
We study conditions for existence, uniqueness and invariance of the comprehensive nonlinear valuation equations first introduced in Pallavicini et al (2011) \cite{pallavicini2011funding}. These equations take the form of semi-linear PDEs and Forward-Backward Stochastic Differential Equations (FBSDEs). After summarizing the cash flows definitions allowing us to extend valuation to credit risk and default closeout, including collateral margining with possible re-hypothecation, and treasury funding costs, we show how such cash flows, when present-valued in an arbitrage free setting, lead to semi-linear PDEs or more generally to FBSDEs. We provide conditions for existence and uniqueness of such solutions in a viscosity and classical sense, discussing the role of the hedging strategy. 
We show an invariance theorem stating that even though we start from a risk-neutral valuation approach based on a locally risk-free bank account growing at a risk-free rate, our final valuation equations do not depend on the risk free rate. Indeed, our final semi-linear PDE or FBSDEs and their classical or viscosity solutions depend only on contractual, market or treasury rates and we do not need to proxy the risk free rate with a real market rate, since it acts as an instrumental variable. 
The equations derivations, their numerical solutions, the related XVA valuation adjustments with their overlap, and the invariance result had been analyzed numerically and extended to central clearing and multiple discount curves in a number of previous works, including \cite{pallavicini2011funding}, \cite{pallavicini2012funding},  \cite{PallaviciniBrigo2013multicurva}, \cite{BrigoPallavicini2014} and \cite{BrigoPallaviciniPedersen}. 
\end{abstract}

\medskip

\medskip

\noindent\textbf{AMS Classification Codes}: 35K58, 60H30, 91B70 \newline


\noindent \textbf{Keywords}: Counterparty Credit Risk, Funding Valuation Adjustment, Funding Costs, Collateralization, Non-linearity Valuation Adjustment, Nonlinear Valuation, Derivatives Valuation, semi-linear PDE, FBSDE, BSDE, Existence and Uniqueness of solutions, Viscosity Solutions.

\pagestyle{myheadings} \markboth{}{{\footnotesize  Brigo, D., Francischello M., and Pallavicini, A. Invariance, Existence and Uniqueness of  Nonlinear Valuation Eqs}}


\section{Introduction}
This is a technical paper where we analyze in detail invariance, existence and uniqueness of solutions for nonlinear valuation equations inclusive of credit risk, collateral margining with possible re-hypothecation, and funding costs.
In particular, we study conditions for existence, uniqueness and invariance of the comprehensive nonlinear valuation equations first introduced in Pallavicini et al (2011) \cite{pallavicini2011funding}. After briefly summarizing the cash flows definitions allowing us to extend valuation to default closeout, collateral margining with possible re-hypothecation and treasury funding costs, we show how such cash flows, when present-valued in an arbitrage free setting, lead straightforwardly to semi-linear PDEs or more generally to FBSDEs. We study conditions for existence and uniqueness of such solutions in a viscosity or classical sense. 
 We formalize an invariance theorem showing that even though we start from a risk-neutral valuation approach based on a locally risk-free bank account growing at a risk-free rate, our final valuation equations do not depend on the risk free rate at all. In other words, we do not need to proxy the risk-free rate with any actual market rate, since it acts as an instrumental variable that does not manifest itself in our final valuation equations. 
Indeed, our final semi-linear PDEs or FBSDEs and their classical or viscosity solutions depend only on contractual, market or treasury rates and contractual closeout specifications once 
we use an hedging strategy that is defined as a straightforward generalization of the natural delta hedging in the classical setting. 

The equations derivations, their numerical solutions and the invariance result had been analyzed numerically and extended to central clearing and multiple discount curves in a number of previous works, including
\cite{pallavicini2011funding}, \cite{pallavicini2012funding},  \cite{PallaviciniBrigo2013multicurva}, \cite{BrigoPallavicini2014},  \cite{BrigoPallaviciniPedersen}, and the monograph \cite{BrigoMoriniPallavicini2012}, which further summarizes earlier credit and debit valuation adjustments (CVA and DVA) results. We refer to such works and references therein for a general introduction to comprehensive nonlinear valuation and to the related issues with valuation adjustments related to credit (CVA), collateral (LVA) and funding costs (FVA).  In this paper, given the technical nature of our investigation and the emphasis on nonlinear valuation,  we refrain from decomposing the nonlinear value into valuation adjustments or XVAs. Moreover, in practice such separation is possible only under very specific assumptions while in general all terms depend on all risks due to non-linearity. Forcing separation may lead to double counting, as initially analyzed through the Nonlinearity Valuation Adjustment (NVA) in \cite{BrigoPallaviciniPedersen}. Separation is discussed in the CCP setting in \cite{BrigoPallavicini2014}. 

The paper is structured as follows.

Section \ref{sec:probabilistic} introduces the probabilistic setting, the cash flows analysis and derives a first valuation equation based on conditional expectations. 
Section \ref{sec:fbsde} derives a FBSDE under the default-free filtration from the initial valuation equation under assumptions of conditional independence of default times and of default-free initial portfolio cash flows. Section \ref{sec:markovian} specifies the FBSDE obtained earlier to a Markovian setting, and derives a semi-linear PDE by assuming regularity of the FBSDE solution. 
Section \ref{sec:fbsdeexun} studies conditions for existence and uniqueness of solutions for the nonlinear valuation FBSDE and classical or viscosity solutions to the associated PDE.
 Section \ref{sec:invariance} presents the invariance theorem:  when adopting delta-hedging, the solution does not depend on the risk-free rate. Section \ref{sec:conclu} re-writes the valuation equation in a way that resembles a risk neutral expectation  while highlighting the key differences and concludes the paper. 

\section{Cash flows analysis and first valuation equation}\label{sec:probabilistic}
We fix a filtered probability space $(\Omega, \AA, \QQ)$, with a filtration $(\GG_u)_{u\geq 0}$ representing the evolution of all the available information on the market. 
With an abuse of notation, we will refer to $(\GG_u)_{u\geq 0}$ by $\GG$. 
The object of our investigation is a portfolio of contracts, or ``contract" for brevity, typically a netting set, with final maturity $T$,
 between two financial entities, the investor $I$ and the counterparty $C$. Both $I$ and $C$ are supposed to
 be subject to default risk. In particular we model their default times with two $\GG$-stopping times $\tau_I,\tau_C$. We assume that the stopping times are generated
by Cox processes of intensities $\lambda^I$ and $\lambda^C$. Furthermore we describe the \emph{default-free} information by means of a 
filtration $(\FF_u)_{u\geq 0}$  generated by the price of the underlying $S_t$ of our contract. This process has the following dynamic under the the measure $\QQ$:
\[
dS_t=r_tS_tdt+\sigma(t,S_t)dW_t
\] 

where $r_t$ is an $\FF$-adapted process, called the \emph{risk-free} rate. We then suppose the existence of a risk-free account $B_t$ following the dynamics
\[
dB_t=r_tB_tdt.
\]
We denote $D(s,t,x)=e^{-\int_s^tx_udu}$ the discount factor associated to the rate $x_u$.
In the case of the risk-free rate we define $D(s,t)\coloneqq D(s,t,r)$.

We further assume that for all $t$ we have $\GG_t=\FF_t\vee \HH^I_t\vee \HH_t^C$ where 
\begin{equation*}
\begin{aligned}
\HH_t^I=\sigma(\II{\tau_I\leq s}, \ s\leq t),\\
\HH_t^C=\sigma(\II{\tau_C\leq s}, \ s\leq t).
\end{aligned}
\end{equation*}
Again we indicate  $(\FF_u)_{u\geq 0}$ by $\FF$ and we will write $\EE^{\GG}_t[\cdot] \coloneqq \EE[\cdot | \GG_t]$ and similarly for $\FF$.
Moreover we postulate the default times to be \emph{conditionally independent} with respect to $\FF$, as in the classic framework of Duffie and Huang \cite{DuffieHuang}, and we indicate $\tau=\tau_I\wedge\tau_C$.
With these assumptions we have that the stopping time $\tau$ has intensity $\lambda_u=\lambda_u^I+\lambda_u^C$. 

For convenience of notation we use the symbol $\bar{\tau}$ to indicate the minimum between $\tau$ and $T$.

\begin{oss}
We suppose that the measure $\QQ$ is the so called \emph{risk-neutral} measure, i.e. a measure under which the prices of the traded non-dividend-paying assets
discounted at the risk-free rate are martingales or, in equivalent terms, the measure associated with the numeraire $B_t$.
\end{oss}

\subsection{The Cash Flows}

To price this portfolio we take the conditional expectation of all the cash flows of the portfolio and discount them at the risk-free rate. An alternative to the explicit cash flows approach adopted here is discussed in \cite{BieleckiRutkowski2014}.

To begin with, we consider a collateralized hedged contract, so the cash flows generated by the contract are: 

\begin{itemize}
\item The payments due to the contract itself, modeled by an $\FF$-predictable process $\pi_t$ of finite variation and a final cash flow
$\Phi(S_T)$ payed at maturity modeled by a Lipschitz function $g$.
At time $t$ the cumulated discounted flows due to these components amount to
\[
\II{\tau>T}D(0,T)\Phi(S_T)+\int_t^{\bar{\tau}}D(t,u)\pi_udu.
\]
\item The payments due to default, in particular we suppose that at time $\tau$ we have a cash flow due to the
default event (if it happened) modeled by a $\GG_\tau$-measurable random variable $\theta_\tau$. So the flows due to this component
are
\[
\II{t<\tau<T}D(t,\tau)\theta_\tau=\II{t<\tau<T}\int_t^TD(t,u)\theta_u d\II{\tau\leq u}.
\]
\item The payments due to the collateral account, more precisely we model this account by a $\FF$-predictable process $C_t$. We postulate that
$C_t>0$ if the investor is the collateral taker, and $C_t<0$ if the investor is the collateral provider. Moreover we assume that
the collateral taker remunerates the account at a certain interest rate (written on the CSA), in particular we may have different rates 
depending on who is the collateral taker, so we introduce the rate
\begin{equation}\label{collateralrate}
c_t=\II{C_t>0}c_t^++\II{C_t\leq 0}c_t^- \ ,
\end{equation}
where $c_t^+,c_t^-$ are two $\FF$-predictable processes. We also suppose that the collateral can be re-hypotecated, i.e. the collateral taker can use the collateral
for funding purposes. Since the collateral taker has to remunerate the account at the rate $c_t$ 
the discounted flows due to the collateral can be expressed as a cost of carry and sum up to
\[
\int_t^{\bar{\tau}}D(t,u)(r_u-c_u)C_udu.
\]
\item We suppose that the deal we are considering is to be hedged by a position in cash and risky assets, represented 
respectively by the $\GG$-adapted processes $F_t$ and $H_t$, with the convention that $F_t>0$ means that the investor is borrowing money
(from the bank's treasury for example), while $F<0$ means that $I$ is investing money. Also in this case to take into account
different rates in the borrowing or lending case we introduce the rate
\begin{equation}\label{fundingrate}
f_t=\II{V_t-C_t>0}f_t^++\II{V_t-C_t\leq 0}f_t^-.
\end{equation}
The flows due to the funding part are
\[
\int_t^{\bar{\tau}}D(t,u)(r_u-f_u)F_udu.
\]
For the flows related to the risky assets account $H_t$, we have that $H_t>0$ means that we need some risky asset,
so we borrow it, while if $H<0$ we lend it. So, for example, if we need to borrow the risky asset we need cash from the treasury,
hence we borrow cash at a rate $f_t$ and as soon as we have the asset we can repo lend it at a rate $h_t$. In general $h_t$ is defined as
\begin{equation}\label{hedgingrate}
h_t=\II{H_t>0}h^+_t+\II{H_t\leq 0}h^-_t.
\end{equation}
Thus we have that the total discounted cash flows for the risky part of the hedge are equal to
\[
\int_t^{\bar{\tau}}D(t,u)(h_u-f_u)H_udu.
\]
\end{itemize}
The last expression could also be seen as resulting from $(r-f) - (r-h)$, in line with the previous definitions. 
If we add all the cash flows mentioned above we obtain that the value of the contract $V_t$ must satisfy
\begin{equation}
\begin{aligned}
V_t=&\EE^{\GG}_t\left[\II{\tau>T}D(t,T)\Phi(S_T)+ \int_t^{\bar{\tau}}D(t,u)(\pi_u+(r_u-c_u)C_u+(r_u-f_u)F_u-(f_u-h_u)H_u)du\right]\\
&+\EE^{\GG}_t\bigg[D(t,\tau)\II{t<\tau<T}\theta_\tau\bigg].
\end{aligned}
\end{equation}
  If we further suppose that we are able to replicate the value of our contract using the funding, the collateral (assuming re-hypothecation, otherwise $C$ is to be omitted from the following equation) and the risky asset accounts, i.e.
\begin{equation}\label{eq:replica}
V_u=F_u+H_u+C_u,
\end{equation}
we have, substituting for $F_u$:
\begin{equation}\label{eq:expected}
\begin{aligned}
V_t=&\EE^{\GG}_t\left[\II{\tau>T}D(t,T)\Phi(S_T)+\int_t^{\bar{\tau}}D(t,u)(\pi_u+(f_u-c_u)C_u+(r_u-f_u)V_u-(r_u-h_u)H_u)du\right]\\
&+\EE^{\GG}_t\bigg[D(t,\tau)\II{t<\tau<T}\theta_\tau\bigg].
\end{aligned}  
\end{equation}

\begin{oss}
In the classic no-arbitrage theory and in a complete market
setting, without credit risk, the hedging process $H$ would correspond 
to a delta hedging strategy account. Here we do not enforce this interpretation yet. 
However we will see that a delta-hedging interpretation emerges from
the combined effect of working under the default-free filtration
$\FF$ (valuation under partial information) and of identifying
part of the solution of the resulting BSDE, under reasonable
regularity assumptions, as a sensitivity of the value to the
underlying asset price $S$.
\end{oss}

\subsection{Adjusted cash flows under a simple trading model}

We now show how the adjusted cash flows originate assuming we buy a call option on an equity asset $S_T$ with strike $K$.
We analyze the operations a trader would enact with the treasury and the repo market in order to fund the
 trade, and we map these operations to the related cash flows. 
We go through the following steps in each small interval $[t,t+dt]$, seen from the point of view of the trader/investor buying the option.
This is written in first person for clarity and is based on conversations with traders working with their bank treasuries.  

\medskip

\noindent Time $t$:
\begin{enumerate}

\item I wish to buy a call option with maturity $T$ whose current price is $V_t = V(t,S_t)$.
I need $V_t$ cash to do that. So I borrow $V_t$ cash from my bank treasury and buy the call.
\item I receive the collateral amount $C_t$ for the call, that I give to the treasury.

\item Now I wish to hedge the call option I bought. To do this, I plan to repo-borrow $\Delta_t = \partial_S V_t$ stock on the repo-market.

\item To do this, I borrow $H_t = \Delta_t S_t$ cash at time $t$ from the treasury.

\item I repo-borrow an amount $\Delta_t$ of stock, posting cash $H_t$ as a guarantee.

\item I sell the stock I just obtained from the repo to the market, getting back the price $H_t$ in cash.
\item I give $H_t$ back to treasury. 
\item My outstanding debt to the treasury is $V_t-C_t$.
\end{enumerate}

\medskip

Time $t+dt$:
\begin{enumerate}
\setcounter{enumi}{8}
\item  I need to close the repo. To do that I need to give back $\Delta_t$ stock. I need to buy this stock from the market. To do that I need $\Delta_t
S_{t+dt}$ cash.
\item I thus borrow $\Delta_t S_{t+dt}$ cash from the bank treasury.
\item I buy $\Delta_t$ stock and I give it back to close the repo and I get back the cash $H_t$ deposited at time
$t$ plus interest $h_tH_t$.
\item I give back to the treasury the cash $H_t$ I just obtained, so that the net value of the repo operation has been
\[ H_t(1+h_t \, dt ) - \Delta_t S_{t+dt} = - \Delta_t \, dS_t + h_t H_t \, dt\]
Notice that this $- \Delta_t dS_t$ is the right amount I needed to hedge
$V$ in a classic delta hedging setting.
\item I close the derivative position, the call option, and get $V_{t+dt}$ cash.

\item I have to pay back the collateral plus interest, so I ask the treasury the amount $C_t(1+c_t\, dt)$ that I give back to the counterparty.

\item My outstanding debt plus interest (at rate $f$) to the treasury is\\ $V_t-C_t+C_t(1+c_t\, dt)+(V_t-C_t)f_t\, dt=V_t(1+f_t\, dt)+C_t(c_t-f_t\, dt)$.\\ I then give to the treasury the cash $V_{t+dt}$ I just obtained,
the net effect being
\[ V_{t+dt} - V_t(1+f_t\, dt)-C_t(c_t-f_t)\, dt = dV_t - f_t V_t \, dt-C_t(c_t-f_t) \, dt\]
\item I now have that the total amount of flows is :
\[ - \Delta_t \, dS_t + h_t H_t \, dt+dV_t - f_t V_t \, dt-C_t(c_t-f_t) \, dt\]
\item Now I present--value the above flows in $t$ in a risk neutral setting.
\[ {\EE}_t [ - \Delta_t \, dS_t + h_t H_t \, dt+dV_t - f_t V_t \, dt-C_t(c_t-f_t) \, dt] =
 - \Delta_t  (r_t - h_t) S_t\, dt + (r_t - f_t) V_t \, dt-C_t(c_t-f_t) \, dt-d \varphi(t) \]
\[ = -H_t(r_t - h_t)\, dt + (r_t - f_t) (H_t+F_t+C_t) \, dt-C_t(c_t-f_t) \, dt-d \varphi(t)\]
\[= (h_t-f_t)H_t\, dt+(r_t - f_t)F_t\, dt+(r_t-c_t)C_t\, dt -d \varphi(t)\]
This derivation holds assuming that ${\EE}_t [ dS_t] = r_t S_t \, dt$ and ${\EE}_t [ dV_t] = r_t V_t \, dt - d \varphi(t)$, where $d \varphi$ is a dividend of $V$ in $[t,t+dt)$ expressing the funding costs. 
Setting the above expression to zero we obtain
\[ d \varphi(t) = (h_t-f_t)H_t\, dt+(r_t - f_t)F_t\, dt+(r_t-c_t)C_t\, dt \]
which coincides with the definition given earlier in \eqref{eq:expected}. 

\end{enumerate}

\section{A FBSDE under $\FF$}\label{sec:fbsde}

We aim to switch to the default free filtration $\FF=(\FF_t)_{t\geq 0}$, and the following lemma
 (taken from Bielecki and Rutkowski \cite{bilrutcredit} Section $5.1$) is the key in understanding how the information 
expressed by $\GG$ relates to the one expressed by $\FF$.  

\begin{lem}\label{lemma:fondamentale}
  For any $\AA$-measurable random variable $X$ and any $t\in \RR_+$, we have:
  \begin{equation}\label{switch}
    \EE_t^{\GG}[\II{t<\tau\leq s}X]=\II{\tau>t}
\frac{\EE_t^{\FF}[\II{\II{t<\tau\leq s}}X]}{\EE_t^{\FF}[\II{\tau>t }]}.
  \end{equation}
  In particular we have that for any $\GG_t$-measurable random variable $Y$ there exists an $\FF_t$-measurable random variable $Z$ such that 
  \[
  \II{\tau>t}Y=\II{\tau>t}Z.
  \]
\end{lem}

What follows is an application of the previous lemma exploiting the fact that we have to deal with a stochastic process structure and not only a simple random variable. Similar 
results are illustrated in \cite{bjrcredit}.

\begin{lem}\label{lemma:scambio}
Suppose that $\phi_u$ is a $\GG$-adapted process. We consider a default time $\tau$ with intensity 
$\lambda_u$. If we denote $\bar{\tau}=\tau\wedge T$ we have:

\begin{equation*}
\EE^{\GG}_t\left[ \int_t^{\bar{\tau}}\phi_udu\right]=\II{\tau>t}\EE^{\FF}_t\left[\int_t^TD(t,u,\lambda)\widetilde{\phi_u}du\right]
\end{equation*}
where $\widetilde{\phi_u}$ is an $\FF_u$ measurable variable such that $\II{\tau>u}\widetilde{\phi_u}=\II{\tau>u}\phi_{u}$.
\end{lem}
\begin{proof}

\[
\EE^{\GG}_t\left[ \int_t^{\bar{\tau}}\phi_udu\right]=\EE^{\GG}_t\left[ \int_t^T\II{\tau>t}\II{\tau>u}\phi_udu\right]=\int_t^T\EE^{\GG}_t\left[\II{\tau>t}\II{\tau>u}\phi_u\right]du
\]
then by using Lemma \ref{lemma:fondamentale} we have
\[
=\int_t^T\II{\tau>t}\frac{\EE^{\FF}_t\left[ \II{\tau>t}\II{\tau>u}\phi_u\right]}{\QQ[\tau>t \ | \FF_t]}du=\II{\tau>t}\int_t^T\EE^{\FF}_t\left[ \II{\tau>u}\phi_u\right]D(0,t,\lambda)^{-1}
du
\]
now we choose an $\FF_u$ measurable variable such that $\II{\tau>u}\widetilde{\phi_u}=\II{\tau>u}\phi_{u}$ and obtain
\begin{equation*}
\begin{aligned}
&=\II{\tau>t}\int_t^T\EE^{\FF}_t\left[\EE^{\FF}_u\left[ \II{\tau>u}\right]\widetilde{\phi_u}\right]D(0,t,\lambda)^{-1}du=\II{\tau>t}\int_t^T\EE^{\FF}_t\left[D(0,u,\lambda)\widetilde{\phi_u}\right]D(0,t,\lambda)^{-1}du\\
&=\II{\tau>t}\EE^{\FF}_t\left[\int_t^TD(t,u,\lambda)\widetilde{\phi}_udu\right]
\end{aligned}
\end{equation*}
\end{proof}

A similar result will enable us to deal with the default cash flow term. In fact we have the following (Lemma 3.8.1 in \cite{bjrcredit})

\begin{lem}\label{lem:localization}
Suppose that $\phi_u$ is an $\FF$-predictable process. We consider two conditionally independent default times
$\tau_I,\tau_C$ generated by Cox processes with $\FF$-intensity rates
$\lambda^I_t,\lambda^C_t$. If we denote $\tau=\tau_C\wedge \tau_I$ we have:
\begin{align*}
\EE^{\GG}_t\left[ \II{t<\tau<T} \II{\tau_I<\tau_C}\phi_\tau\right]=\II{\tau>t}\EE_t^\FF\left[\int_t^TD(t,u,\lambda^I+\lambda^C)\lambda^I_u\phi_udu\right].
\end{align*}
\end{lem}

Now we postulate a particular form for the default cash flow, more precisely if we indicate $\widetilde{V}_t$ the $\FF$-adapted process such that
\[
\II{\tau>t}\widetilde{V}_t=\II{\tau>t}V_t
\]
then we define
\[
\theta_t=\epsilon_{t}-\II{\tau_C<\tau_I}LGD_C(\epsilon_{t}-C_t)^++\II{\tau_I<\tau_C}LGD_I(\epsilon_{t}-C_t)^-.
\]
Where $LGD$ indicates the loss given default, typically defined as $1-REC$, where $REC$ is the corresponding recovery rate and $(x)^+$ indicates the positive part of $x$ and $(x)^-=-(-x)^+$. The meaning of these flows is the following, consider $\theta_\tau$: 
\begin{itemize}
\item at first to default time $\tau$ we compute the close-out value $\epsilon_{\tau}$;
\item if the counterparty defaults and we are net debtor, i.e. $\epsilon_{\tau}-C_\tau\leq 0$ then we have to pay the whole close-out value $\varepsilon_\tau$ to the counterparty;
\item if the counterparty defaults and we are net creditor, i.e. $\epsilon_{\tau}-C_\tau>0$ then we are able to recover just a fraction of our credits, namely 
$C_\tau+REC_C(\varepsilon_\tau-C_\tau)=REC_C\varepsilon_\tau+LGD_CC_\tau=\varepsilon_\tau-LGD_C(\varepsilon_\tau-C_\tau)$ where $LGD_C$ indicates the loss given default and is equal to 
one minus the recovery rate $REC_C$.
\end{itemize}
A similar reasoning applies to the case when the Investor defaults.

If we now change filtration, we obtain the following expression for $V_t$ (where we omitted the tilde sign over the rates, see Remark~\ref{oss:Frates}):
\begin{equation}
\begin{aligned}
V_t=&\II{\tau>t}\EE^{\FF}_t\left[ D(t,T,r+\lambda)\Phi(S_T)+\int_t^{T}D(t,u,r+\lambda)(\pi_u+(f_u-c_u)C_u+(r_u-f_u)\widetilde{V}_u-(r_u-h_u)\widetilde{H}_u)du\right]\\
&+\II{\tau>t}\EE^{\FF}_t\left[\int_t^TD(t,u,r+\lambda)\widetilde{\theta}_u du\right],\\
\end{aligned}  
\end{equation}
where, if we suppose $\epsilon_t$ to be $\FF$-predictable, we have (using Lemma \ref{lem:localization}): 

\begin{equation}
\begin{aligned}
\widetilde{\theta}_u&=\epsilon_{u}\lambda_u-LGD_C(\epsilon_{u}-C_u)^+\lambda^C_u+LGD_I(\epsilon_{u}-C_u)^-\lambda^I_u.
\end{aligned}  
\end{equation}

\begin{oss}\label{oss:Frates}
From now on we will omit the tilde sign over the rates $f_u,h_u$. Moreover we note that if a rate is of the form
\[x_t=x^+\II{g(V_t,H_t)>0}+x^-\II{g(V_t,H_t)\leq 0}\]
then on the set $\{\tau>t\}$ it coincides with the rate
\[
\widetilde{x}_t=\widetilde{x}^+\II{g(\widetilde{V}_t,\widetilde{H}_t)>0}+\widetilde{x}^-\II{g(\widetilde{V}_t,\widetilde{H}_t)\leq 0}.
\]
\end{oss}

We note that this expression is of the form $V_t=\II{\tau>t}\Upsilon$ meaning that $V_t$ is zero on $\{\tau\leq t\}$ and that on the set $\{\tau>t\}$ it coincides with
 the $\FF$-measurable random variable $\Upsilon$. But we already know a variable that coincides with $V_t$ on $\{\tau>t\}$, i.e. $\widetilde{V}_t$. Hence we can write the following

\begin{equation}\label{eq:Vtilde}
\begin{aligned}
\widetilde{V}_t=&\EE^{\FF}_t\left[ D(t,T,r+\lambda)\Phi(S_T)+\int_t^{T}D(t,u,r+\lambda)(\pi_u (f_u-c_u)C_u+(r_u-f_u)\widetilde{V}_u-(r_u-h_u)\widetilde{H}_u)du\right]\\
&+\EE^{\FF}_t\left[\int_t^TD(t,u,r+\lambda)\widetilde{\theta}_u du\right].\\
\end{aligned}  
\end{equation}

We now show a way to obtain a BSDE from equation \eqref{eq:Vtilde}, another possible approach (without default risk) is shown for example in \cite{RutkowskiNie2014}.
We introduce the process

\begin{equation}
\begin{aligned}
  X_t=&\int_0^{t}D(0,u,r+\lambda)\pi_u du+\int_0^tD(0,u,r+\lambda)\widetilde{\theta}_udu\\
&+\int_0^{t}D(0,u,r+\lambda)\left[(f_u-c_u)C_u+(r_u-f_u)\widetilde{V}_u-(r_u-h_u)\widetilde{H}_u\right]du.
\end{aligned}
\end{equation}

Now we can construct a martingale summing up $X_t$ and the discounted value of the deal as in the following:
\[
D(0,t,r+\lambda)\widetilde{V}_t+X_t=\EE^{\FF}_t[X_T+D(0,T)\Phi(S_T)].
\]
So differentiating both sides we obtain:
\[
-(r_u+\lambda_u)D(0,u,r+\lambda)\widetilde{V}_udu+D(0,u,r+\lambda)d\widetilde{V}_u+dX_u=d\EE^{\FF}_u[X_T+D(0,T)\Phi(S_T)]
\]
If we substitute for $X_t$ we have: 
\[
d\widetilde{V}_u+\left[\pi_u-(r_u+\lambda_u)\widetilde{V}_u+\widetilde{\theta}_u+(f_u-c_u)C_u+(r_u-f_u)\widetilde{V}_u-(r_u-h_u)\widetilde{H}_u\right]du=\frac{d\EE^{\FF}_u[X_T+D(0,T)\Phi(S_T)]}{D(0,u,r+\lambda)}
\]
The process $(\EE^{\FF}_t[X_T+D(0,T)\Phi(S_T)])_{t\geq 0}$ is clearly a closed $\F$-martingale, and hence $\int_0^tD(0,u,r+\lambda)^{-1}d\EE^{\FF}_u[X_T+D(0,T)\Phi(S_T)]$ is a local $\F$-martingale. 
Then, being $\int_0^tD(0,u,r+\lambda)^{-1}d\EE^{\FF}_u[X_T+D(0,T)\Phi(S_T)]$ adapted to the Brownian driven filtration $\FF$, by the martingale representation theorem
 we have  $\int_0^tD(0,u,r+\lambda)^{-1}d\EE^{\FF}_u[X_T+D(0,T)\Phi(S_T)]=\int_0^tZ_udW_u$ for some $\FF$-predictable process $Z_u$. Hence we can write:
\begin{equation}\label{diffV}
d\widetilde{V}_u+\left[\pi_u-(f_u+\lambda_u)\widetilde{V}_u+\widetilde{\theta}_u+(f_u-c_u)C_u-(r_u-h_u)\widetilde{H}_u\right]du=Z_udW_u
\end{equation}

\section{Markovian FBSDE and PDE for $\widetilde{V}_t$}\label{sec:markovian}

As it is, equation \eqref{diffV} is way too general, thus we will make some simplifying assumptions in order to guarantee existence and uniqueness of a solution. First we assume a Markovian setting, and hence we suppose that all the processes appearing in \eqref{diffV} are deterministic functions of $S_u,\widetilde{V}_u$ or $Z_u$ and time.  
More precisely we assume that:
\begin{itemize}
\item the dividend process $\pi_u$ is a deterministic function $\pi(u,S_u)$ of $u$ and $S_u$, Lipschitz continuous in $S_u$;
\item the rates $r,f^\pm,c^\pm,\lambda^I,\lambda^C,h^\pm$ are deterministic bounded functions of time.
\item the collateral process is a fraction of the process $\widetilde{V}_{u}$, namely $C_u=\alpha_u\widetilde{V}_{u}$, where $0\leq\alpha_u\leq 1$ is a function of time;
\item the close-out value $\epsilon_t$ is equal to $\widetilde{V}_t$ (this adds a source of non-linearity with respect to choosing a risk-free closeout, see for example \cite{BrigoMoriniPallavicini2012} and \cite{BrigoPallaviciniPedersen});
\item the hedging process is of the form
  $\widetilde{H}_u=H(u,S_u,\widetilde{V}_u,Z_u)$, where $H(u,s,v,z)$ is a
  deterministic function Lipschitz-continuous in $v,z$ uniformly in $u$;
\item the diffusion coefficient $\sigma(t,S_t)$ of the underlying dynamic is Lipschitz continuous uniformly in time in $S_t$
\end{itemize}
\begin{oss}
Note that under certain assumptions on the coefficients of the dynamics of $\widetilde{V}_t$, in Section \ref{sec:fbsdeexun} 
we will actually show that $\widetilde{V}_t$ is a continuous process and hence a predictable one, 
so that the assumptions on the collateral and close-out value processes are reasonable.
\end{oss}

\begin{oss}
The reason why we can postulate such a specific form for the hedging process is that since the default intensities are deterministic, 
the only risk left to be hedged once we switched to filtration $\FF$ is the market risk.
\end{oss}

Under our assumptions, equation \eqref{diffV} becomes the following FBSDE (that has been rewritten to emphasize the dependence on the initial data,
 and without the tildes to ease the notation):

\begin{equation}\label{flowFBSDE}
  \begin{aligned}
&dS^{q,s}_t=r_tS^{q,s}_tdt+\sigma(t,S^{q,s}_t)dW_t \quad q<t\leq T \\\\
&S_q=s_q \quad 0\leq t \leq q \\\\
&dV^{q,s}_t=-\underbrace{\left[\pi_t+\theta_t(f_t(\alpha_t-1)-\lambda_t-c_t\alpha_t)V^{q,s}_t-(r_t-h_t)H(t,S^{q,s}_t,V^{q,s}_t,Z^{q,s}_t)
\right]}_{B(t,S^{q,s}_t,V^{q,s}_t,Z^{q,s}_t)}dt+Z^{q,s}_tdW_t\\
&V^{q,s}_T=\Phi(S_T^{q,s}).
  \end{aligned}
\end{equation}


Now we wish to explain intuitively how we can obtain a Black-Scholes like PDE from our FBSDE. 
Let us assume that the value of our contract is a deterministic, $C^{1,2}$ function of time and of the underlying, i.e.
$V_t^{q,s}=u(t,S^{q,s}_t)$. Then we can write the Ito's formula for $u(t,S^{q,s}_t)$, obtaining:
\begin{equation}\label{Ito}
du(t,S^{q,s}_t)=\left(\partial_tu(t,S^{q,s}_t)+r_tS^{q,s}_t\partial_su(t,S^{q,s}_t)+\frac{1}{2}\sigma(t,S^{q,s}_t)^2\partial^2_su(t,S^{q,s}_t)\right)dt+\sigma(t,S^{q,s}_t)\partial_su(t,S^{q,s}_t)dW_t.
\end{equation}
Then by comparing expressions \eqref{Ito} and \eqref{flowFBSDE} we have the following

\begin{equation}\label{comparison}
  \begin{aligned}
&\partial_tu(t,S^{q,s}_t)+ r_tS^{q,s}_t\partial_su(t,S^{q,s}_t)+\frac{1}{2}\sigma(t,S^{q,s}_t)^2\partial^2_su(t,S^{q,s}_t)=-B(t,S^{q,s}_t,\widetilde{V}_t,Z_t^{q,s})\\
&\sigma(t,S^{q,s}_t)\partial_su(t,S^{q,s}_t)=Z_t^{q,s} .
  \end{aligned}
\end{equation}
So, $V_t^{t,s}=u(t,s)$ satisfies the following semilinear PDE:

\begin{equation}\label{PDEV}
  \begin{aligned}
&\partial_tu(t,s)+\frac{1}{2}\sigma(t,s)^2\partial^2_su(t,s)+\mu(t,s)\partial_su(t,s)+B(t,s,u(t,s),(\partial_s u \sigma)(t,s))=0\\
&u(T,s)=\Phi(s)
  \end{aligned}
\end{equation}
Moreover we see from \eqref{comparison} that the process $Z_t^{q,s}$ is in a certain sense, a multiple of the delta-hedging process.

\section{FBSDE Existence and Uniqueness Results}\label{sec:fbsdeexun}
We now state the precise conditions under which we can obtain existence and uniqueness of the solution to both the FBSDE and the PDE of the previous section.
More specifically as done in  Pardoux and Peng \cite{PardouxPeng1992} we have (for a generalization to the case of fully coupled FBSDE see for example \cite{Delarue2002}) :

\begin{teo}\label{Papeng}
Consider the following FBSDE on the interval $[0,T]$

\begin{equation}\label{genericFBSDEdecoupled}
  \begin{aligned}
&dX^{q,x}_t=\mu(t,X^{q,x}_t)dt+\sigma(t,X^{q,x}_t)dW_t \quad q<t\leq T \\
&X_t=x \quad 0\leq t\leq q\\
&dY^{q,x}_t=-f(t,X^{q,x}_t,Y^{q,x}_t,Z^{q,x}_t)dt+Z^{q,x}_tdW_t\\
&Y^{q,x}_T=g(X^{q,x}_T)
  \end{aligned}
\end{equation}

Assume that there exist a constant $K$ such that $\forall t$
\begin{itemize}
  \item  $|\mu(t,x)-\mu(t,x')|+|\sigma(t,x)-\sigma(t,x')|\leq K|x-x'|$
  \item  $|\mu(t,x)|+|\sigma(t,x)|\leq K(1+|x|)$
  \item  $|f(t,x,y,z)-f(t,x,y',z')|\leq K(|y-y'|+|z-z'|)$
\end{itemize}
Moreover suppose that there exist a constant $p\geq 1/2$ such that:
\[|g(x)|+|f(t,x,0,0)|\leq K(1+|x|^p)\]
and that the map
\[x\mapsto(f(t,x,0,0),g(x))\]
is continuous,then there exist two measurable deterministic functions $u(t,x),\ d(t,x)$ such that 
the unique solution $(X^{q,x}_t,Y^{q,x}_t,Z^{q,x}_t)$ of \eqref{genericFBSDEdecoupled} is given by
\[
Y^{q,x}_t=u(t,X_t^{q,x}) \qquad Z^{q,x}_t=d(t,X_t^{q,x})\sigma(t,X_t^{q,x})
\]
and moreover  $u(t,x)=Y^{t,x}_t$ is the unique viscosity solution to following PDE 

\begin{equation}\label{PDE}
  \begin{aligned}
&\partial_tu(t,x)+\frac{1}{2}\sigma(t,x)^2\partial^2_xu(t,x)+\mu(t,x)\partial_xu(t,x)+f(t,x,u(t,x),\sigma(t,x)\partial_xu(t,x))=0\\
&u(T,x)=g(x)
  \end{aligned}
\end{equation}

\end{teo} 

In order to have a classical solution to equation \eqref{PDE} we need to assume some smoothness of the coefficients of equation \eqref{genericFBSDEdecoupled}.
A possible choice is the following (see J.Zhang \cite{ZhangPhdthesis} Theorem 2.4.1 on page 41):
\begin{teo}\label{teo:classicsolution}
Consider Equation \eqref{genericFBSDEdecoupled}. If we assume that there exists a positive constant $K$ such that
\begin{itemize}
\item $\sigma(t,x)^2\geq \frac{1}{K}$;
\item  $|f(t,x,y,z)-f(t,x',y',z')|+|g(x)-g(x')|\leq K(|x-x'|+|y-y'|+|z-z'|)$;
\item $|f(t,0,0,0)|+|g(0)|\leq K$;
\end{itemize}
and moreover the functions $\mu(t,x)$ and $\sigma(t,x)$ are $C^2$ with bounded derivatives, then
equation \eqref{genericFBSDEdecoupled} has a unique solution $(X^{q,x}_t,Y^{q,x}_t,Z^{q,x}_t)$ 
and  $u(t,x)=Y^{t,x}_t$ is the unique \emph{classical} (i.e. $C^{1,2}$) solution to the semilinear PDE \eqref{PDE}.
\end{teo}

Our aim is applying Theorem \ref{Papeng} to our FBSDE. Indeed, one can check that the assumptions of Theorem \ref{Papeng} are satisfied in our setting, so as to obtain the following theorem on existence and uniqueness of a classical solution of the semilinear PDE for Credit-Collateral-Funding-closeout inclusive valuation.

\begin{teo}\label{Lipschitz_viscosity}(Existence \& uniqueness of viscosity solution of semilinear PDE for comprehensive valuation).    
If the rates $\lambda_t,\ f_t,\ c_t,\ h_t,\ r_t$ are bounded, then $|B(t,s,v,z)-B(t,s',v',z')|\leq K(|v-v'|+|z-z'|)$.
Hence if there exists a $p\geq 1/2$ such that $|B(t,s,0,0)|+\Phi(s)\leq K(1+|s|^p)$ the assumptions of Theorem \ref{Papeng} are satisfied and so
 equation \eqref{flowFBSDE} has a unique solution, and moreover
 $u(t,s)=V_t^{t,s}$ is a viscosity solution to the following semilinear PDE:

\begin{equation}\label{PDEdependent}
\begin{aligned}
&\partial_tu(t,s)+\frac{1}{2}\sigma(t,s)^2\partial^2_su(t,s)+r_ts\partial_su(t,s)+B(t,s,u(t,s),\sigma(t,s)\partial_su(t,s))=0\\
&u(T,s)=\Phi(s)
\end{aligned}
\end{equation} 

\end{teo} 

 \begin{proof}
We start by rewriting the term 
\[
B(t,\omega,v,z)=\pi_t(s)+\theta_t(v)+(f_t(\alpha_t-1)-\lambda_t-c_t\alpha_t)v-(r_t-h_t)H(u,s,v,z).
\]
Since the sum of two Lipschitz functions is itself a Lipschitz function we can restrict ourselves to analyzing the summands that appear in the previous
formula.
The term $\pi_t$ is Lipschitz continuous in $s$ by assumption. The $\theta$ term and the $(f_t(\alpha_t-1)-\lambda_t-c_t\alpha_t)v$ term are continuous and
 piece-wise linear, hence Lipschitz continuous. The last term is piece-wise linear as a function of $H$ which is a Lipschitz function of $v,z$.
\end{proof}

\section{Invariance Theorem}\label{sec:invariance}
We now want to specialize equation \eqref{flowFBSDE} to the case in which we use
delta-hedging. In particular following the heuristic reasoning in Section
\ref{sec:markovian} we choose
\[
\widetilde{H}_t=S_t\frac{Z_t}{\sigma(t,S_t)}.
\]

Now we prove that with this choice equation \eqref{flowFBSDE} has a
solution $V_t^{q,s}$ such that $V_t^{q,s}=u(t,S_t^{q,s})$ with $u\in C^{1,2}$. We cannot directly apply Theorem \ref{teo:classicsolution} to our FBSDE because $B(t,s,v,z)$ is not Lipschitz continuous in $s$ because of the hedging term. But, since the delta-hedging term is linear in $Z_t$ we can move it from the drift of the backward equation to the drift of the forward one. More precisely consider the following: 

\begin{equation}\label{FBSDE_h}
  \begin{aligned}
&dS^{q,s}_t=h_tS^{q,s}_tdt+\sigma(t,S^{q,s}_t)dW_t \quad q<t\leq T \\\\
&S_q=s_q \quad 0\leq t \leq q \\\\
&dV^{q,s}_t=-\underbrace{\left[\pi_t+\theta_t-\lambda_tV^{q,s}_t+f_tV^{q,s}_t(\alpha_t-1)-c_t(\alpha_tV^{q,s}_t)\right]}_{B'(t,S^{q,s}_t,V^{q,s}_t)}dt+Z^{q,s}_tdW_t\\
&V^{q,s}_T=\Phi(S_T^{q,s}).
  \end{aligned}
\end{equation}

Note that the $S$-dynamics in \eqref{FBSDE_h} has the repo rate $h$ as drift .
Since in general $h$ will depend on the future values of the deal, thiscould be a source of nonlinearity and is at times represented informally with an expected value $\mathbb{E}^h$ or a pricing measure $\mathbb{Q}^h$, see for example \cite{BrigoPallaviciniPedersen} and the related discussion on operational implications for the case $h=f$. 
Indeed, one can check that the assumptions of Theorem \ref{teo:classicsolution} are satisfied for this equation:

\begin{teo}\label{Lipschitz}    
If the rates $\lambda_t,\ f_t,\ c_t,\ h_t,\ r_t$ are bounded, then $|B'(t,s,v)-B'(t,s',v')|\leq K(|s-s'|+|v-v'|)$ and $|B'(t,0,0)|+\Phi(0)\leq K$.
Hence if $\sigma(t,s)$ is a positive $C^2$ function with bounded derivatives and
the rate $h_t$ does not depend on the sign of $H$, namely $h^+=h^-$, then the assumptions of Theorem \ref{teo:classicsolution} are satisfied and so
 equation \eqref{FBSDE_h} has a unique solution, and moreover $V_t^{t,s}=u(t,s)\in C^{1,2}$ and satisfies the following semilinear PDE:

\begin{equation}\label{PDEindependent}
\begin{aligned}
&\partial_tu(t,s)+\frac{1}{2}\sigma(t,s)^2\partial^2_su(t,s)+h_ts\partial_su(t,s)+B'(t,s,u(t,s))=0\\
&u(T,s)=\Phi(s)
\end{aligned}
\end{equation} 

\end{teo}

We now show that a solution to equation \eqref{flowFBSDE} can be obtained by
means of the classical solution to the PDE \eqref{PDEdependent}. We start considering the following forward equation which is known to have a unique solution under our assumptions about $\sigma(t,s)$. 
\begin{equation}\label{forward}
dS_t=r_tS_tdt+\sigma(t,S_t)dW_t \quad S_0=s.
\end{equation}
We define $V_t=u(t,S_t)$ and $Z_t=\sigma(t,S_t)\partial_su(t,S_t)$. By Theorem \ref{Lipschitz} we know that $u(t,s)\in C^{1,2}$ and by applying Ito's formula and \eqref{PDEindependent} we obtain:
\begin{equation*}
\begin{aligned}
dV_t=&du(t,S_t)=\left(\partial_tu(t,S_t)+r_tS_t\partial_su(t,S_t)+\frac{1}{2}\sigma(t,S_t)^2\partial^2_su(t,S_t)\right)dt+\sigma(t,S_t)\partial_su(t,S_t)dW_t\\
&=\left((r_t-h_t)S_t\partial_su(t,S_t)-B'(t,S_t,u(t,S_t))\right)dt+\sigma(t,S_t)\partial_su(t,S_t)dW_t\\
&=\left((r_t-h_t)S_t\frac{Z_t}{\sigma(t,S_t)}-\pi_t(S_t)-\theta_t(V_t)-(f_t(\alpha_t-1)-\lambda_t-c_t\alpha_t)V_t)\right)dt+Z_tdW_t
\end{aligned}
\end{equation*}
Hence we found the following:
\begin{teo}[Solution to the Valuation Equation]\label{teo:existence}
Let $S_t$ be the solution to equation \eqref{forward} and $u(t,s)$ the classical solution to equation \eqref{PDEdependent}. Then the process $(S_t,u(t,S_t),\sigma(t,S_t)\partial_su(t,S_t))$ is the unique solution to equation \eqref{flowFBSDE}.
\end{teo}
\begin{proof}
	From the reasoning above we found that
  $(S_t,u(t,S_t),\sigma(t,S_t)\partial_su(t,S_t))$ solves the equation
  \eqref{flowFBSDE}. Then from Theorem \ref{Lipschitz_viscosity} we know that equation
  \eqref{flowFBSDE} has a unique solution and hence we have the thesis.
\end{proof}

\begin{oss}\label{oss:delta}
	Since we proved that $V_t=u(t,S_t)$ with $u(t,s)\in C^{1,2}$, the reasoning we used, when saying that $\widetilde{H}_t=S_t\frac{Z_t}{\sigma(t,S_t)}$ represented choosing a delta-hedge, it's actually more than an heuristic argument. 
\end{oss}

  Moreover since \eqref{PDEindependent} does not depend on the risk-free rate $r_t$, we can state the following:

\begin{teo}[Invariance Theorem]
If we are under the assumptions of Theorem \ref{Lipschitz}  and we assume that we are backing our deal with a delta
hedging strategy, then the price $V_t$ can be calculated via the semilinear PDE \eqref{PDEindependent} and does \emph{not} depend on the risk-free rate $r(t)$.

Moreover if in analogy with the just mentioned classical case we choose $H(t,s,u(t,s),Z_t)=S_t\frac{Z^{t,s}_t}{\sigma(t)}$ under the weaker assumptions of Theorem \ref{Lipschitz_viscosity}, we still have that the price is a viscosity solution of equation \eqref{PDEindependent} and hence does \emph{not} depend on the risk-free rate $r(t)$.
\end{teo}

This invariance result shows that even when starting from a risk neutral valuation theory, the risk free rate disappears from the nonlinear valuation equations. 

\section{Conclusions: Nonlinear deal-dependent  measures and discouting}\label{sec:conclu}
Using a Feynman-Kac type argument formally, but with the confidence coming from the above result on existence and uniqueness of  solutions,  we can also write the valuation formula as
\[ \hspace{-1cm} {V}_t
 =  \int_t^T \,\EE^h\{  D(t,u;f) [ \pi_u + (\tilde{\theta}_u - \lambda_u V_u) + ( {f}_u - c_u ) C_u ] |{\cal F}_t \} du . \]
Related formulas were introduced in previous papers such as \cite{pallavicini2011funding} and \cite{pallavicini2012funding}.
While this formula stays as close as possible to classical risk neutral valuation, we can see immediately where we depart from the usual setting. $\EE^h$ is the expectation associated with $\mathbb{Q}^h$,  the probability measure where the drift of the risky assets is the repo rate $h$. This repo rate depends on $H$ and hence on $V$ itself. This confirms nonlinearity and can be further interpreted as a deal-dependent pricing measure. The pricing measure depends on whether the repo will be long or short in the future, as rates $h$ could be different in the two cases, and on the specific repo portfolio adopted for the trade under consideration. This is visible also in (\ref{FBSDE_h}), where the drift of $S$ is $h$. Furthermore, we ``discount at funding". Note that ${f}$ depends on ${V}$ possibly. This is another potential source of non-linearity, that is here interpreted as ``nonlinear discounting". In other terms we have a  deal dependent discount curve.   
We recall that $\theta_u$ are trading CVA and DVA after collateralization and can be nonlinear under replacement closeout.
Finally, $( {f}_u - {c}_u ) C_u$ is the cost of funding collateral with the treasury, and can be nonlinear as well. We have been able to assert an invariance theorem, confirmed by the fact that also in the valuation formula with $\EE^h$ there is no risk-free rate $r$, but we cannot avoid the nonlinearities in case of asymmetric borrowing/lending rates or in case of replacement closeout at default. In case linearization is enforced, the related error should be controlled with quantities related to the nonlinearity valuation adjustment (NVA) introduced in \cite{BrigoPallaviciniPedersen}.
Further discussion on consequences of non-linearity and invariance on valuation in general, on the operational procedures of a bank, on the legitimacy of fully charging the nonlinear value to a client and on the related dangers of overlapping valuation adjustments is presented elsewhere, see for example again \cite{BrigoPallavicini2014}, \cite{BrigoPallaviciniPedersen} and references therein.

\nocite{}
\bibliographystyle{plain}

\end{document}